\documentclass{article}
 
\usepackage{authblk}
\usepackage{amsmath,amsthm,amssymb}
\usepackage{fullpage}
\usepackage{xcolor}
\usepackage{enumitem}
\usepackage[hidelinks]{hyperref}
\usepackage{textcomp}
\usepackage{url}
\usepackage{natbib}

\newtheorem{theorem}{Theorem}[section]

\newtheorem{lemma}[theorem]{Lemma}

\newtheorem{remark}[theorem]{Remark}
\newtheorem{definition}{Definition}[section]

\begin{document}

\title{Procurement Auctions with Best and Final Offers\thanks{Vasilis Gkatzelis was on sabbatical at Google Research in NYC when this work took place.}}

\author[1]{Vasilis Gkatzelis\thanks{gkatz@drexel.edu}} 
\author[2]{Randolph Preston McAfee\thanks{randolph@mcafee.cc}}
\author[2]{Renato Paes Leme\thanks{renatoppl@google.com}}

\affil[1]{Drexel University, Computer Science}
\affil[2]{Google Research}

\date{}
\maketitle

\begin{abstract}
We study sequential procurement auctions where the sellers are provided with a ``best and final offer'' (BAFO) strategy. This strategy allows each seller $i$ to effectively ``freeze'' their price while remaining active in the auction, and it signals to the buyer, as well as all other sellers, that seller $i$ would reject any price lower than that. This is in contrast to prior work, e.g., on descending auctions, where the options provided to each seller are to either accept a price reduction or reject it and drop out. As a result, the auctions that we consider induce different extensive form games and our goal is to study the subgame perfect equilibria of these games. We focus on settings involving multiple sellers who have full information regarding each other's cost (i.e., the minimum price that they can accept) and a single buyer (the auctioneer) who has no information regarding these costs. Our main result shows that the auctions enhanced with the BAFO strategy can guarantee efficiency in every subgame perfect equilibrium, even if the buyer's valuation function is an arbitrary monotone function. This is in contrast to prior work which required that the buyer's valuation satisfies restrictive properties, like gross substitutes, to achieve efficiency. We then also briefly analyze the seller's cost in the subgame perfect equilibria of these auctions and we show that even if the auctions all return the same outcome, the cost that they induce for the buyer can vary significantly. 
\end{abstract}

\section{Introduction}
Many situations call for purchasing goods or services that need to ``fit together.''  A computer drive must fit in the drive bay, an alternator in the engine bay, and a piece of accounting software may need to be compatible with inputs from a web store.  Indeed, a company setting up a web store needs a variety of software products that work together: a front end for the buyer, an accounting system, an inventory system, a process for product delivery, customer service software, and so on.  The creator of a web store will want to purchase various applications that are compatible and, as a result, two accounting systems may be good substitutes, but an accounting system will be a complement to an inventory system.  That is, an array of software applications will contain both substitutes and complements, in a potentially very complex relationship.

Procurement problems featuring both substitutes and complements are probably more common than just substitutes.  Virtually any assembly problem will force complements via specifications and compatibility.  While individual components are substitutes for each other, components that fit together are necessarily complementary.  This logic applies equally well to software as well as hardware.  Similarly, information -- for advertising, for financial transactions, for investment, or for training models -- from various sources will have both complements and substitutes among the various suppliers.  How should a buyer purchase in such a situation?

The default solution for pricing an asset when there is substantial uncertainty about its value is to hold an auction, and this approach has been successful for a variety of such assets, ranging from electromagnetic spectrum licenses to antiques and art. 
We approach this problem from the buyer's perspective, aiming to design procurement (or reverse) auctions that they can use to determine which subset of these goods or services they should acquire, and at what price (e.g., the US government uses online reverse auctions to procure goods and services in a competitive environment~\cite{ProcurementGov}). One difficulty in applying standard procurement auction tools to the applications discussed above is that known techniques from the literature focus on the case where  the goods being sold are substitutable, e.g., captured by the assumption that the buyer's valuation function over subsets of the goods is submodular, or satisfies the more restrictive property of gross substitutes. For example, the brilliant paper of \citet{KC82}, extended by \citet{GS99}, shows that whenever the gross substitutes property is satisfied the auctions perform well. Moreover, there are examples that illustrate bad equilibria can arise whenever the gross substitutes property is violated, suggesting that this property may even be necessary.

Our main contribution in this paper is to design auctions that can handle the more complicated value structure that arises in many important applications. We revisit the design of procurement auctions and augment the bidder with a new type of strategy which we call ``best and final offer'' (BAFO). This strategy, which is often used in practice but has not received enough attention in auction theory, effectively allows the bidders to freeze their price, while remaining active. In doing so, they risk losing the auction but, as we show, if they strategically choose when to use it, the design recovers the efficiency even in settings with highly complicated valuation functions that exhibit both substitutes and complements.

\subsection{Our Results}

We focus on settings where a buyer wants to acquire the goods or services of multiple sellers. The buyer's preferences are characterized by a (potentially highly complicated) combinatorial valuation function which assigns a numerical value to each subset of goods or services, exhibiting both substitutes and complements. On the other hand, each seller has a cost for selling, i.e., a minimum payment that they would require in order to sell. Our goal is to design procurement auctions that identify which sellers to buy from and at what price, leading to efficient outcomes.
To that end, we analyze two classes of sequential procurement auctions that augment the bidder with a new type of strategy: the ``best and final offer'' (BAFO). 

In the first class of auctions, which we refer to as ``Name-Your-BAFO'' auctions, the buyer interacts with each seller only once: the sellers are approached in some order and each seller $i$ is asked to directly report the payment they request for their goods or services as a bid $b_i \in \mathbb{N}$. Before choosing this bid, each seller can observe all previously reported bids. Once every seller has reported their bid, the buyer determines which subset of sellers to purchase from, aiming to maximize his utility (i.e., his value for the services minus their cost). If the buyer ends up purchasing from some seller $i$, then the seller's payment is equal to their bid $b_i$. Therefore, when the seller reports their bid, they are essentially offering a take-it-or-leave-it price to the buyer and, since they do not interact after that, this bid is their best and final offer.

The second class of auctions are descending auctions and may interact multiple times with each seller, providing them with a strategy space that is quite different. These auctions initially assign a high price to each seller and then take place over a sequence of rounds: in each round the auction computes a tentative allocation based on the current prices and approaches one of the sellers that remain active and is not in the tentative allocation. This seller is asked if they would be willing to accept a slightly reduced price, and they can choose to either accept that reduction or to permanently ``freeze'' their price. Crucially, a seller that freezes their price does not drop out of the auction; they instead use the freezing strategy to signal to the buyer and the other sellers that they are not willing to accept any lower price, and they commit to this signal. Once again, the freezing strategy corresponds to a best and final offer from the seller, who then just waits until the end of the auction for the buyer to determine the set of winners. Once every seller is either frozen or in the tentative allocation, the auction terminates and the buyer purchases from the set of sellers that maximizes their utility at the final prices. The crucial difference between this class of auctions and other descending auctions (e.g., descending clock auctions) is that rejecting a price decrease does not automatically cause a seller to drop out of the auction.

\paragraph{Analysis of the auction efficiency and cost.} Since these auctions are sequential (i.e., they take place over a sequence of rounds), we model them as extensive form games and we analyze the subgame perfect equilibria (SPE) of these games. We approach this problem from the perspective of the buyer (or the auctioneer) who we assume has no information regarding the sellers' true costs (i.e., the minimum price that they would require in order to provide their goods or services) and uses the auction as a way to determine which subset of sellers to buy from and at what price. On the other hand, we assume that the sellers know each other's costs and are competing with each other. The buyer's valuation function (i.e., what subset of services they would buy at a given set of prices) is public information.

Our main result shows that any sequential procurement auction from the aforementioned two classes is guaranteed to reach efficient outcomes in every subgame perfect equilibrium. This result holds for a very general class of buyer valuation functions that can exhibit both complementarities and substitutes, which is in stark contrast to classic prior work that requires the restrictive gross substitutes property in order to reach efficiency. The key differences that allow us to achieve this positive result are i) the fact that our auctions provide the sellers with the ability to freeze their price without dropping out (the BAFO strategy), combined with ii) their sequential implementation, which allows the sellers to signal to each other using the BAFO strategy. Specifically, once some seller has finalized their price using the BAFO strategy, all other active sellers can observe this fact and have to choose their optimal strategy conditioned on this fact. This is something that each seller anticipates when choosing the BAFO strategy, and this signaling between the sellers avoids inefficient outcomes.

Finally, apart from the allocations returned by these auctions we also analyze the price vectors that they give rise to. We show that although the allocation remains efficient irrespective of which auction is used, the auction choice can have a very big impact on the total cost that the buyer needs to pay, i.e., the sum of the prices that they need to pay in the efficient solution.

\subsection{Related Work}
Our work adds a new twist to the large body of literature on ascending and descending price auctions, which have a long history in economic theory (e.g., \citet{KC82}, \citet{demange1986multi}; \citet{GS99}; \citet{parkes2000iterative}; \citet{ausubel2002ascending}; \citet{bikhchandani2002package}, \citet{ausubel2004efficient, ausubel2006efficient}; \citet{perry2005efficient}; \citet{de2007ascending}). Such designs have been immensely successful both in theory and in practice and variants of these designs have been used in major spectrum auctions worldwide (e.g., \citet{ausubel2014market,MS20}) as well as in auctions for electricity, gas, and
emission allowances in Europe (\citet{cramton2012discrete}), among many other applications. There are several practical advantages to this auction format such as minimizing the information the buyer learns about the sellers and the simplicity in bidding (see \citet{perry2005efficient} for a comprehensive discussion).

Prior work has focused on settings in which there is a degree of substitution between the items involved: \citet{KC82,GS99} focus on gross substitutes, \citet{demange1986multi} on unit-demands which is a special case of substitutes, \citet{ausubel2004efficient} on homogenous goods with decreasing marginals, and \citet{de2007ascending} assumes a submodularity condition. Two examples that we are aware of with ascending and descending price procedures that can handle complements are: \citet{sun2006equilibria} and \citet{baranov2017efficient} who do so by studying restricted forms of complementarity with an underlying substitutable structure. This is a structure that  \citet{hatfield2016hidden} refer to as ``hidden substitutes''. Some of these auctions are known to lose efficiency even in seemingly very simple settings beyond substitutes \cite{DGR17}.

Just like the descending auctions that we analyze in this paper, the descending clock auctions in the papers cited above also assign a personalized price to each seller, which then weakly decreases over time. However, a crucial difference is that clock auctions do not provide the bidders with the option of freezing their price. If a seller is not willing to accept a price decrease, then they are forced to drop out of the auction and are, therefore, excluded from the final solution, even if their price before dropping out turns out to be competitive in hindsight.

Our results augment the descending auction format with the ability for participants to commit to a BAFO signal -- instead of dropping out of the auction, sellers can remain active but can no longer revise their price. This new feature allows us to extend the efficiency guarantees to any combinatorial valuation. The notion of a BAFO has been used in the implementation of optimal strategies in certain bargaining games (e.g. Samuelson \cite{samuelson1984bargaining}) but to the best of our knowledge has not been applied to iterative combinatorial auctions. 
One exception is some ongoing work on budget-feasible mechanism design by~\citet{AranyakBAFO}.

Another important difference between our work and the previous literature on iterative auctions is the equilibrium concept. Traditional descending price auctions satisfy the stronger notion of strategyproofness. Instead we guarantee efficiency under any subgame perfect equilibrium (SPE) of the extensive form game induced by the auction mechanism. Characterizing SPE tends to be difficult except for very structured games \citet{leme2012curse}. The SPE of simple auction formats has been studied under submodularity and matroid conditions \citet{leme2012sequential}.

\section{Preliminaries}
We consider settings with a single buyer who wants to procure goods or services from a set $N=\{1,\hdots, n\}$ of $n$ sellers. Each seller $i\in N$ has a cost $c_i\in \mathbb{N}$ for selling\footnote{Throughout the paper, we assume that costs and prices are expressed as multiples of some small enough denomination, e.g., \$1 or \textcent 1.} and the buyer's value for buying each subset of goods or services is captured by a combinatorial valuation function $v(\cdot): 2^N \to \mathbb{R}$. Given a vector of prices $p \in \mathbb{R}^n$, one for each seller, the utility of the buyer for a subset of sellers $Q$ is equal to his value for that subset, minus the total cost, i.e., $v(Q)-\sum_{i\in Q} p_i$. We use $D(v;p)$ to denote the buyer's preferred subset of sellers, i.e., the one that maximizes the buyer's utility:
$$D(v;p) := \text{argmax}_{Q \subseteq N} [v(Q) - \textstyle\sum_{i \in Q} p_i].$$

\paragraph{Types of valuation functions.}
A valuation function $v(\cdot):2^N \to \mathbb{R}$ is \emph{submodular} if for any two sets $Q$ and $R$ such that $Q\subset R \subset N$, and any $i\notin Q\cup R$, we have $v(Q\cup \{i\})-v(Q) \geq v(R\cup \{i\})-v(R)$. A valuation function satisfies the \emph{gross substitutes} property if an increase of the price for certain goods does not reduce the buyer's demand for goods whose price did not increase. Formally, $v(\cdot)$ satisfies \emph{gross substitutes} if for every pair of price vectors $p \leq p'$, if $Q \in D(v;p)$ at prices $p$, then there exists $R \in D(v;p')$ at prices $p'$ such that $Q \cap \{i\in N: p_i= p'_i\} \subseteq R$. A valuation function is \emph{anonymous} if the value $v(Q)$ of every set $Q\subseteq N$ depends only on the size $|Q|$ of this set (e.g., this type of valuation would arise if every seller is offering the same type of good or service).

\paragraph{Procurement auctions.}
Our goal is to design procurement auctions, which are mechanisms that interact with the sellers and decide which subset  $W\subseteq N$ of sellers the buyer will purchase from (the ``winners''), and what payment each winner $i\in W$ should receive. The output of the auction is the set $W$ (or, equivalently, a vector $x\in \{0, 1\}^n$, such that $x_i=1$ if $i\in W$ and $x_i =0$ if $i\notin W$) and a price $p_i$ for each $i\in W$. The utility of each seller $i\in N$ is 
\[u_i(x, p) =(p_i - c_i)  x_i.\] 
If $p_i = c_i$ we assume that agents prefer winning at a price equal to their cost to losing (both yield the same utility of zero).

Given a set of winners $W$ and a price vector $p$, the social welfare is equal to the sum of the utilities of the buyer and the sellers, i.e., 
\[\left(v(W)-\sum_{i\in W} p_i \right)+ \sum_{i\in W} (p_i - c_i) ~=~ v(W)-\sum_{i\in W} c_i.\] 
We use $W^*\in \arg\max_{W\subseteq N}v(W)-\sum_{i\in W} c_i$ to denote an \emph{efficient} solution, i.e., a set of winners $W$ that maximizes the social welfare, and we use $x^*$ to denote the corresponding allocation vector. The total cost of the buyer given a price vector $p$ and a set of winners $W$ is $\sum_{i\in W} p_i$.

\paragraph{Winner selection and tie-breaking.} Once the price vector $p$ has been finalized, our proposed auctions choose a set of winners $W$ from the demand set $D(v;p)$. If there are multiple such sets in $D(v;p)$, we tie break using a winner selection rule 
which maps each price vector $p$ to a set of winners $W(p) \in D(v;p)$ and satisfies the well-known Independence of Irrelevant Alternatives (IIA) property. In our context, this implies that if $W^*$ is the winning set that we choose from $D(v;p)$ and, $D(v,p')$ is another demand set such that $D(v,p')\subseteq D(v,p)$ and $W^*\in D(v,p')$, then we choose $W^*$ from $D(v,p')$ as well. Most natural ways of breaking ties satisfy IIA (e.g., any rule that defines some arbitrary total order over subsets $Q \subseteq N$ and then  chooses the first $Q\in D(v;p)$ according to this total order).

\newcommand{\T}{\mathcal{T}}
\newcommand{\R}{\mathbb{R}}
\renewcommand{\S}{\mathcal{S}}

\paragraph{Extensive form games and subgame perfect equilibrium.} Our auctions are sequential and give rise to extensive form games between the sellers. 
An extensive form game with a set $N$ of players is represented by a rooted tree $\T$  of finite depth with node set $\S$. Given a node $s \in \S$, we use $\T(s)\subseteq \S$ to denote the child-nodes of $s$ in $\T$. If $\T(s) = \emptyset$, we say that $s$ is a terminal node; otherwise, we say it is an internal node. We use $\S_{\textsf{term}}$ to denote the set of terminal nodes and $\S_{\textsf{int}}=\S \setminus \S_{\textsf{term}}$ to denote the set of internal nodes.

Each internal node $s$ is associated with a player $i(s)$ (the player whose turn it is to make a ``move'' at that point in the game) using a mapping $i:\S_{\textsf{int}} \rightarrow N$. If we let $\S_i\subseteq \S_{\textsf{int}}$ denote the set of internal nodes associated with each player $i$, then a strategy for player $i$ is a mapping $a_i: \S_i \rightarrow \S$ such that $a_i(s) \in \T(s)$. Each terminal node is associated with a payoff $\pi_j: \S_{\textsf{term}} \rightarrow \R$ for each player $j\in N$. Given a profile of strategies $(a_1,\hdots, a_n)$, one for each player, we can also define payoff functions $\hat{\pi}_j$ over all nodes, $\S$, using the terminal node payoffs $\pi_j$ and backward induction, as follows: 
$$\hat{\pi}_j(s) = \pi_j(s) \text{ for } s \in \S_{\textsf{term}} \text{ and } \hat{\pi}_j(s) = \hat{\pi}_j(a_{i(s)}(s)) \text{ for } s \in \S_{\textsf{int}}.$$

A profile of strategies is a subgame perfect equilibrium (SPE) if for every player $i$, all nodes $s \in \S_i$ and all nodes $s' \in \T(s)$ we have:
$$ \hat{\pi}_i(a_{i}(s)) \geq \hat{\pi}_i(s'),$$
i.e., no player $i$ can unilaterally change their strategy at $s$ from $a_i(s)\in \T(s)$ to some other $s'\in \T(s)$ and increase their utility.

\section{Sequential Name-Your-BAFO auction}

As a warm up, we first propose and analyze the \emph{Name-Your-BAFO} auction. This auction interacts with each seller only once, asking them to report their best and final offer, so it is simpler to implement and analyze than the descending auction of Section~\ref{sec:descending}, but the latter allows for a gradual ``price discovery,'' which makes descending auctions attractive in practice. These differences notwithstanding, we show that they both always yield efficient outcomes.

The Name-Your-BAFO auction approaches the sellers in some (possibly adaptive) order and sequentially asks each seller $i\in N$ to submit a bid $b_i \in \mathbb{N}$ regarding the payment that they request for their good or service.  Before reporting their bid, each seller can observe the bids reported by all preceding sellers. Once every seller has reported their bid, the buyer then chooses a set of winners $W$ that maximizes their utility $v(W)-\sum_{i\in W}b_i$ (using a tie-breaking rule that satisfies IIA) and pays each winner $i\in W$ their bid, i.e., $p_i = b_i$. Since the auction can only choose a seller if it pays them their bid, this bid corresponds to the seller's best and final offer. 

The auction described above gives rise to an extensive form game that can be represented by a tree of depth $n+1$, where each node at level $\ell+1$ is indexed by a tuple $(b_1, \hdots, b_\ell) \in \mathbb{N}^\ell$, corresponding to the preceding bids. Note that the tree has infinite branching but finite depth. The payoffs of the sellers can be computed at the terminal nodes once all the bids $(b_1, \hdots, b_n)$ have been specified. In this section we only refer to this tree implicitly. In the more difficult proofs in the following session, we explicitly analyze the corresponding game tree. 

Before analyzing this auction, we consider the toy example of buying chopsticks in auction to develop some intuition regarding the important role of the sequential nature of our auctions.

\subsection{Buying Chopsticks in Auction}
To exhibit the issues that arise with complementarities and how sequential pricing can sidestep these issues, we consider the illustrative example of buying chopsticks using an auction. Consider an instance of an auction involving one seller with a fork and two sellers with one chopstick each. The value of the buyer for a fork is the same as the value for two chopsticks, say \$1; a single chopstick is worthless. This situation fails the gross substitutes property because the chopsticks are complements (raising the price of one chopstick will eventually lower the demand for the other chopstick). Assume that the chopstick sellers have a cost of \textcent10 each, and the fork seller has a cost of \textcent50, and consider a non-sequential version of the Name-Your-BAFO auction. In this auction, there is an equilibrium where the chopstick sellers each ask for \textcent95 and the buyer buys the fork for \$1. No seller can improve their position by a unilateral deviation; a chopstick seller cannot create a sale even if they were to lower their price to equal their cost, i.e., \textcent10.

Note, however, that the inefficient equilibrium falls apart if the prices are set sequentially and the sellers can observe these prices.  For example, suppose that the auction approaches the sellers in a sequence and the fork seller goes last. If the sum of the chopstick bids exceeds the buyer's value of \$1, the fork seller can ask for \$1 and win the auction. If the sum of the chopstick bids does not exceed \$1, but is more than the fork seller's cost of \textcent50, the fork seller can ask for the sum of the chopstick bids, perhaps minus a penny, and win the auction. Knowing this, the second chopstick seller will ensure that the sum of the prices does not exceed the fork's cost, if that is possible (i.e., if the price asked by the first seller is not more than \textcent40).  Knowing this, the first chopstick asks for exactly \textcent40, the second seller asks for \textcent10, and the fork seller asks for \textcent50.\footnote{This assumes that the tie-breaking favors the chopsticks.If this is not the case, then the first seller would just bid \textcent39 instead, ensuring that the cost of chopsticks is \textcent 49.} The other cases, where the fork seller is not last, are similar but the fork seller's bid may not necessarily equal their cost.  In all permutations, the buyer purchases chopsticks, which is the efficient outcome, though the total payment may exceed \textcent50.

\begin{remark}
As exhibited by the chopstick auction example above, in the subgame perfect equilibria of the Name-Your-BAFO auction the sellers' bids are not necessarily equal to their costs.
\end{remark}

\subsection{Efficiency of Name-Your-BAFO Auctions}

Our main result for this section shows that every subgame perfect equilibrium of a Name-Your-BAFO auction is efficient.

\begin{theorem}\label{thm:NYBAFO-efficient}   
For any combinatorial valuation $v$, the allocation induced by any subgame perfect equilibrium of a Name-Your-BAFO auction is efficient.
\end{theorem}

Before proving this theorem, we use the fact that all auctions in this paper choose an allocation that maximizes the buyer's utility (using a tie-breaking rule that satisfies IIA) to show that the outcome always satisfies the following property.
\begin{lemma}\label{lem:winner-choice}
If for some prices $p$ our auction chooses the winning set $W(p)$, and $p'$ are any prices such that $p'_i \leq p_i$ for $i \in W(p)$ (the prices of all the winning bidders are weakly lower) and $p'_i \geq p_i$ for $i \notin W(p)$ (the prices of all the losing bidders are weakly higher), then our auction chooses the same winning set at $p'$, i.e., $W(p') = W(p)$.

\end{lemma}
\begin{proof}
Let $W^* = W(p)$ be the set of winners at prices $p$ and let $T \neq W^*$ be any other set of sellers. Since $W^*\in D(v;p)$, we have
\begin{equation}\label{ineq:opt_choice}
v(W^*) - \sum_{i \in W^*} p_i \geq v(T) - \sum_{i \in T} p_i.
\end{equation}
Under prices $p'$, the cost of $W^*$ decreases by $\Delta= \sum_{i\in W^*}p_i -p'_i$ and the cost of every other set $T$ decreases\footnote{In fact, it may even increase.} by at most $\sum_{i\in T\cap W^*} p_i - p'_i$, which is at most $\Delta$, i.e., $\sum_{i\in T} p'_i \geq \sum_{i\in T} p_i -\Delta$. As a result, we get:
\begin{align*}
v(W^*) - \sum_{i \in W^*} p'_i &= v(W^*) +\Delta - \sum_{i \in W^*} p_i \\
                        &\geq v(T) + \Delta - \sum_{i \in T} p_i\\
                        &\geq v(T) - \sum_{i \in T} p'_i,
\end{align*}
where the first inequality is due to~\eqref{ineq:opt_choice} and the second inequality holds because prices in $T$ drop by at most $\Delta$. Therefore, $W^*$ is also in $D(v;p')$ (the demand set for prices $p'$) and $D(v;p')\subseteq D(v;p)$ (the demand set for $p'$ contains only the sets of sellers that were in $D(v;p)$ and whose price dropped by exactly $\Delta$). Since our auctions use a tie-breaking rule that satisfies the IIA property and $W^*$ was selected from $D(v;p)$, it will also be selected from $D(v;p')$.
\end{proof}

We now use this lemma to analyze the subgame perfect equilibria of the game induced by Name-Your-BAFO, which can be computed using backwards induction. First, we focus on the bid of the last seller, given the prices posted by the $n-1$ bidders preceding it. Then, we focus on the bid of the next-to-last seller, and so on. 

For the last seller, choosing their optimal BAFO is relatively straightforward, since all the other bidders' prices are already finalized. The bidder chooses their bid $b_n$ aiming to maximize their utility $(b_n-c_n)x_n$, where $x_n = 1$ if the last bidder is in the winning set and $x_n=0$ otherwise. Note that, according to the definition of the auction, if the bidder wins, the price that they are paid is $p_n=b_n$, but the winning set $W(p)$ is chosen based on the final prices. Hence, the value of $b_n$ that maximizes the bidder's utility, depends on the previously posted bids $(b_1,...,b_{n-1})$.

To simplify our analysis for sellers that arrive earlier in the ordering, we now define a ``conditional price vector'' for each round of the Name-Your-BAFO auction.
\begin{definition}
For each round $k\in \{0, 1, \dots, n\}$ of the Name-Your-BAFO auction, the \emph{conditional price vector} is defined as:
\begin{equation*}
\hat{p}_i(k) = 
\begin{cases}
    b_i, & \text{if $i\leq k$}\\
    c_i, & \text{otherwise.}
\end{cases}
\end{equation*}
The final price vector after the completion of the auction is $p=\hat{p}_n$.
\end{definition} 

Now, using this price vector, we define a notion of ``conditional efficiency.'' The auction directly associates a winning set $W(s)$
with each terminal node $s$ of the game tree (the set of winners if that is indeed the outcome of the auction). Given a subgame perfect equilibrium, this allocation can also be extended to every internal node $s$, using the subgame perfect equilibrium outcome of the subgame rooted at $s$ (i.e., $W(s)$ is the set of winners in the subgame perfect equilibrium outcome of the subtree rooted at $s$).

\begin{definition} Consider a node $s$ at the $k$-th level of the Name-Your-BAFO game tree, after the first $k-1$ sellers have named their best and final offers $b_1, \hdots, b_{k-1}$. We say that the set $W(s)$ associated with $s$ is conditionally efficient if it maximizes the buyer's utility with respect to the price vector $\hat{p}(k-1)$, i.e., 
\begin{equation*}    
W(s) \in \arg\max_X v(X) - \sum_{i\in X} \hat{p}_i(k-1).
\end{equation*}
If there are multiple such allocations, $W(s)$ is the one chosen by the same tie-breaking rule used by the auction to determine the winners.
\end{definition}

We are now ready to prove Theorem~\ref{thm:NYBAFO-efficient}.
\begin{proof}[Proof of Theorem~\ref{thm:NYBAFO-efficient}] We prove that in any subgame perfect equilibrium, the set $W(s)$ associated with each node $s$ of the Name-Your-BAFO game tree is conditionally efficient. We start at level $k=n$ and proceed by backwards induction. This implies that the auction is also conditionally efficient at $k=1$, for which $\hat{p}_i=c_i$ for all $i\in N$, so the Name-Your-BAFO auction is efficient in every SPE. 

The rest of the proof verifies the conditional efficiency for all $k$.

\textbf{Base case ($k=n$):} Given any node $s$ at level $k=n$ of the game tree, we consider two possibilities based on whether or not there is a bid $b_n\geq c_n$ that bidder $n$ can make to become a winner given the previously posted bids $(b_1, \dots, b_{n-1})$. If no such bid exists, then bidder $n$ is not in $W(s)$ (i.e., they are not one of the winners in the SPE outcome of the subgame rooted at $s$) because for this bidder to become a winner they would have to report a bid $b_n <c_n$ that would return negative utility $b_n - c_n$. For every other bidder and, hence, also for every $i\in W(s)$, we have that their conditional price is the same as their final price, i.e., $\hat{p}_i(n-1)= p_i$, and thus $W(s)$ is conditionally efficient. If, on the other hand, there exists some bid $b_n\geq c_n$ that would make bidder $n$ a winner, then the bidder's optimal strategy is to report such a bid in order to be one of the winners (in fact, the bidder's optimal strategy is to report the largest $b_n$ for which they remain a winner). This means that that bidder $n$ will be a winner in the SPE of the subgame rooted at $s$, so this bidder is in $W(s)$. Furthermore, this implies that $W(s)$ is conditionally efficient in this case as well: since i) $W(s)$ maximizes the buyer's utility with respect to the final prices $p$ and ii) the prices $\hat{p}(n-1)$ satisfy $\hat{p}_i(n-1)\leq p_i$ for all $i\in W(s)$ and $\hat{p}_i(n-1)= p_i$ for all $i\notin W(s)$, we can use Lemma~\ref{lem:winner-choice} to conclude that $W(s)$ would also maximize the buyer's utility with respect to $\hat{p}(n-1)$.

\textbf{Induction step:} Now, consider any level $k<n$ and assume that for any subgame perfect equilibrium, any level $\ell \in \{k+1, \dots, n\}$, and any node $s'$ at level $\ell$, we have that $W(s')$ is conditionally efficient. We consider any node $s$ at level $k$ and we consider two possibilities based on whether or not there is a bid $b_k\geq c_k$ that bidder $k$ can make to become a winner in the corresponding child node $s'$ (the child node of $s$ that corresponds to choosing strategy $b_k$) given the previously posted bids $(b_1, \dots, b_{k-1})$. 

If no such bid exists, then bidder $k$ is not in $W(s)$ because for this bidder to become a winner they would have to report a bid $b_k <c_k$ that would return negative utility $b_k - c_k$. For every other bidder and, hence, also for every $i\in W(s)$, we have that their level-$k$ conditional price in $s$ and their level-($k+1$) conditional price in any child node of $s$ is the same (either equal to their cost or equal to their final price), and thus $W(s)$ is conditionally efficient. 

If, on the other hand, there exists some bid $b_k\geq c_k$ that would make bidder $k$ a winner, then the bidder's optimal strategy is to report such a bid in order to be one of the winners. This means that that bidder $k$ will be a winner in the SPE of the subgame rooted at $s$, so this bidder is in $W(s)$. Furthermore, this implies that $W(s)$ is conditionally efficient in this case as well. Let $s'$ be the child node of $s$ that bidder $k$ chooses in the SPE and note that i) $W(s)$ maximizes the buyer's utility with respect to the level-($k+1$) conditional prices at node $s'$ and ii) level-($k+1$) conditional price of bidder $k$ at node $s'$ is weakly greater than the level-$k$ conditional price of bidder $k$ at node $s$ (because the former is $b_k$ while the latter is equal to $c_k$, and $b_k\geq c_k$). All other prices are the same so, using Lemma~\ref{lem:winner-choice} we conclude that $W(s)$ would also maximize the buyer's utility with respect to $\hat{p}(k)$.
\end{proof}

\begin{remark}\label{remark:nyb-order-independent}
The proof of Theorem \ref{thm:NYBAFO-efficient} does not make any assumptions regarding the order in which the bidders are approached, so it holds for any possible, even adaptive, ordering. Also, the proof shows that the SPE strategy in each round depends only on which sellers already reported their bid and what that bid was. A remarkable property of this game, is that sellers can compute their strategies without knowing the order in which the remaining sellers will be approached.
\end{remark}

\begin{remark}
The proof of Theorem~\ref{thm:NYBAFO-efficient} actually provides a partial computation for equilibria.  First, for any $k$ such that $x^*_k=1$, where $x^*$ maximizes $v(x)-cx$, $k$ will choose the highest price for which  $\max_x v(x)-(p_1,\dots,p_{k-1},p_k,c_{k+1},\dots,c_n)x$ results in $x_k=1$.  When $x^*_k=0$, one choice that always results in an equilibrium is $p_k=c_k$, but there could be higher prices that also result in equilibria.
\end{remark}

\section{Descending Auctions with BAFO}\label{sec:descending}

While the Name-Your-BAFO auction has efficient equilibria, it may be a demanding auction format from the perspective of the sellers in practice. They interact with the auction only once and, in this one interaction, they need to choose one of infinitely many strategies: a bid $b_i \in \mathbb{N}$. As an alternative format, we now present a class of descending auctions where the sellers have repeated interactions with the auction as it gradually reduces the prices offered to each of them, and in each interaction a seller needs to choose between just two strategies: to either accept a price decrease or to permanently ``freeze'' their price. The option for a seller to freeze their price is a feature that is often used in practice, but has not received as much attention from an analytical standpoint. By freezing their price, the seller effectively communicates to the buyer that they would not accept any price reductions, so this is their best and final offer.

\paragraph{Descending auctions with BAFO}
The auction first initializes the price of each seller $i\in N$ to $p_i\gets h$ (where $h\in \mathbb{N}$ is some large value; e.g., larger than the costs of all sellers), and it initializes the set $F$ of ``frozen'' sellers to be empty. 
The auction then takes place over a sequence of rounds, and in each round it computes a (possibly empty) tentative allocation $W(p)$ based on the current prices, it chooses a seller $i \notin W(p)\cup F$ 
(i.e., one who has not frozen their price), and it provides them with two options: (i) accept a decrease of their price\footnote{Note that we assumed that all costs are natural numbers, using some small enough monetary denomination, like \$1 or \textcent1, so we can safely restrict our attention to price decrements of $1$. We could, alternatively, just choose a small enough decrement of $\epsilon$.} from $p_i$ to $p_i \gets p_i - 1$, or (ii) freeze their price at $p_i$. If $i$ freezes at $p_i$, this is their best and final offer (BAFO); the bidder is added to $F$ and the auction never attempts to lower their price again in the future. Each bidder's decision to accept a reduced price or freeze is observed by every other bidder, so both the current prices and the subset of bidders whose prices are frozen (i.e., the set $F$) are public knowledge. There are no restrictions on the order in which the sellers are approached. In particular, it can be adaptive, i.e., dependent of the sellers' choices along the way. The auction terminates after all sellers not in the tentative allocation $W(p)$ have frozen their prices. We assume that if a price reaches zero, the seller automatically freezes. Once all prices are finalized, the set of winners is determined to be the set that maximizes the buyer's utility at the given prices. If there are multiple such sets, the auction uses a tie-breaking rule that satisfies IIA.

\paragraph{Extensive form game tree} Each descending auction from the class above induces an extensive form game which can be represented as a tree, with each node corresponding to the interaction of the auction with some seller $i$, asking them to either accept a reduction of their price by $1$ or permanently freeze it. Note that in each round of this game one of the prices is either frozen or decreased by $1$ and no price can drop below zero, so the game is finite. For each node $s$ of this tree, we use $p(s)$ to denote the price vector at the time when the auction reaches this node $s$ and we use $F(s)$ to denote the set of sellers that have already chosen to freeze, i.e., they chose to freeze at some node on the path from the root to node $s$. Note that $p(s)$ and $F(s)$ are both fully determined by the path from the root to $s$, as the edges of this path determine when each price is decreased or frozen. If, at node $s$, the seller $i$ who was asked to reduce their price accepts this reduction, then we proceed to the left child-node of $s$, which we denote as $s_\ell$. If $i$ does not accept this reduction and instead chooses to freeze, then we proceed to the right child-node of $s$, which we denote by $s_r$.
Also, let $G(s)$ be the subgame that corresponds to the substree rooted at $s$.

Note that, since the sequence in which the sellers are approached by the descending auction can depend in non-trivial ways on the observed strategic choices, this can lead to an unpredictable trajectory for the price vector. At first glance, this suggests that the induced game would be very demanding for the sellers to play, but our key lemma (Lemma~\ref{lem:induct_eff}) shows that sellers do not need to know anything regarding the future price trajectory or details regarding the history to compute their optimal strategy. They only need access to $p(s)$ and $F(s)$.

\subsection{Efficiency of Descending Auctions w/ BAFO} 
Given a descending auction and some problem instance, consider any subgame perfect equilibrium of the game tree induced by this auction. We annotate each node $s$ in the game tree with the final allocation $W(s)$ resulting from playing the given subgame perfect equilibrium of the subtree rooted at $s$. 

\begin{definition}\label{def:cost}
For each node $s$ in the descending auction game tree, we define a price vector $\hat{p}(p(s), F(s))$, or just $\hat{p}(s)$, such that
\begin{equation*}
\hat{p}_i(s) = 
\begin{cases}
    p_i(s), & \text{if $i\in F(s)$}\\
    c_i, & \text{otherwise.}
\end{cases}
\end{equation*}
\end{definition} 

We show the following lemma using backwards induction: 
\begin{lemma}\label{lem:induct_eff}
For every node $s$ of the descending auction game tree, the allocation $W(s)$ resulting as the subgame-perfect equilibrium of $G(s)$ is an allocation that maximizes the buyer's utility with respect to cost vector $\hat{p}$, i.e., an allocation
\begin{equation*}
W(s)\in \arg\max_X v(X) - \sum_{i \in X} \hat{p}_i(s).
\end{equation*}
If there are multiple such allocations, $W(s)$ is chosen using the same tie-breaking rule that the auction uses to determine the set of winners.
\end{lemma}
\begin{proof}
To verify that the lemma holds for every terminal node of the game tree, we observe that at every terminal node $s$ all prices have been permanently frozen, so the vector $p$ of prices is finalized. The buyer's chosen allocation in response to these prices would be the allocation $X$ that maximizes $v(X)-\sum_{i\in X}p_i(s)$. Since every seller is frozen at $s$, i.e., $F(s)=N$, we have $\hat{p}_i(s)=p_i(s)$ for all $i$, so the chosen allocation also maximizes $v(X)-\sum_{i\in X}\hat{p}_i(s)$. If there are multiple such allocations, our auction is designed to consistently tie-break, so the lemma holds for all leaves. 

Now, consider an internal node $s$ of the game tree and assume that the lemma holds for all of its descendants. Let $i$ be the seller who, at node $s$, is asked to either accept a price decrease from $p_i(s)$ to $p_i(s)-1$ (leading to child-node $s_\ell$) or to freeze at $p_i(s)$ (leading to child-node $s_r$). By our inductive assumption, if $i$ accepts the price decrease, the resulting allocation will be 
\begin{equation}\label{alloc_n_ell}    
W(s_\ell)\in\arg\max_X v(X) - \sum_{i \in X} \hat{p}_i(s_\ell),
\end{equation}
and, if $i$ chooses to freeze, the resulting allocation will be
\begin{equation}\label{alloc_s_r}    
W(s_r)\in\arg\max_X v(X) - \sum_{i \in X} \hat{p}_i(s_r);
\end{equation}
in both cases any ties are broken using the same tie-breaking rule. 

Now, to identify seller $i$'s optimal strategy at node $s$ we consider four cases based on whether $i$ wins or loses in the aforementioned allocations $W(s_\ell)$ and $W(s_r)$:

\begin{itemize}[leftmargin=*]
\item \textbf{Case 1: Agent $i$ loses at $W(s_\ell)$.} Note that if this is the case, then seller $i$ also loses at $W(s_r)$. To verify this fact, note that $\hat{p}_i(s_\ell)=p_i(s)-1$ and $\hat{p}_i(s_r)=p_i(s)$, so $\hat{p}_i(s_r)>\hat{p}_i(s_\ell)$, while $\hat{p}_j(s_r)=\hat{p}_j(s_\ell)$ for all other sellers $j\neq i$. Therefore, for every winner $j\in W(s_\ell)$ we have $\hat{p}_j(s_r)=c_j(s_\ell)$ and every loser $j\notin W(s_\ell)$ we have $\hat{p}_j(s_r)\geq c_j(s_\ell)$. From the statement of the lemma we can conclude that $W(s_r)=W(s_\ell)$, which means that the allocation $W(s)$ resulting from playing the subgame-perfect equilibrium of $G(s)$ is independent of $i$'s strategic choice at node $s$ and, hence, $W(s)=W(s_\ell)=W(s_r)$. 

Since $W(s)=W(s_\ell)$, to prove that the lemma holds for node $s$ as well, it suffices to show that
\begin{equation}\label{case1}
W(s_\ell)\in \arg\max_X v(X) - \sum_{i \in X} \hat{p}_i(s),  
\end{equation}
and that this allocation would be chosen in case of ties. To verify that both of these are true, note that $\hat{p}_j(s)=\hat{p}_j(s_\ell)$ for every seller $j$, since no additional freezing took place between $s$ and $s_\ell$. Therefore, our inductive assumption that the lemma holds for node $s_\ell$ directly implies that the lemma also holds for $s$.

\item \textbf{Case 2: Agent $i$ wins at both $W(s_\ell)$ and $W(s_r)$.} If this is the case, then seller $i$ would choose to freeze, since they would win in both cases (due to the inductive hypothesis), but the price that $i$ would receive if they freeze is higher (it will be exactly $p_i(s)$ if they freeze, while it would be at most $p_i(s)-1$ if they do not). Therefore, given this strategic choice of $i$, the resulting allocation $W(s)$ of $G(s)$ will be the same as $W(s_r)$. 

Since $W(s)=W(s_r)$, to prove that the lemma holds for $s$ as well, it suffices to show that
\begin{equation}\label{case2}
W(s_r)\in \arg\max_X v(X) - \sum_{i \in X} \hat{p}_i(s),  
\end{equation}
and that $\tau$ would choose this allocation in case of ties. We can verify that~\eqref{case2} holds by using the fact that $W(s_r)$ satisfies condition \eqref{alloc_s_r} and then observing that the only difference between $\hat{p}(s_r)$ and $\hat{p}(s)$ is the fact that $\hat{p}_i(s_r)\geq c_i(s)$ because in $s_r$ seller $i$ froze at a price at least $c_i$. Therefore, for each set that does not contain $i$ their ``cost'' relative to $\hat{p}$ is the same between $s$ and $s_r$, while the ``cost'' relative to $\hat{p}$ of all sets that include $i$ dropped by the same amount. The fact that $i\in W(s_r)$ implies that $W(s_r)$ satisfies~\eqref{case2}. To also verify that $\tau$ would choose $W(s_r)$ in case of ties with respect to $\hat{p}(s)$, note that $c_j(s_r)\leq c_j(s)$ for all $j\in W(s_r)$ and $c_j(s_r)\geq c_j(s)$ for all $j\notin W(s_r)$, so the fact that $\tau$ chose $W(s_r)$ given $\hat{p}(s_r)$ (by our inductive assumption) implies that $\tau$ would also choose $W(s_r)$ given $\hat{p}(s)$ (by definition of the tie-breaking rule). This implies that the lemma holds also for $s$.

\item \textbf{Case 3A: Agent $i$ wins at $W(s_\ell)$ for a price less than $c_i$ and loses at $W(s_r)$.} In this case, if seller $i$ accepted the price decrease they would end up winning, but for a price that is strictly lower than their cost, leading to negative utility. They would instead prefer to freeze at price $p_i(s)$ and lose in order to maintain a non-negative utility, so the resulting allocation is $W(s)=W(s_r)$.

Since $W(s)=W(s_r)$, to prove that the lemma holds for $s$ as well, it suffices to show that
\begin{equation*}
W(s_r)\in \arg\max_X v(X) - \sum_{i \in X} \hat{p}_i(s),  
\end{equation*}
and that the tie-breaking rule would choose this allocation in case of ties. We can once again verify that this is true using the fact that $W(s_r)$ satisfies condition \eqref{alloc_s_r}, combined with the facts that $\hat{p}_i(s_r)\geq \hat{p}_i(s)$ for seller $i$ who is not in $W(s)$, while $\hat{p}_j(s_r)=\hat{p}_j(s)$ for all other sellers $j\neq i$.
\end{itemize}

\item \textbf{Case 3B: Agent $i$ wins at $W(s_\ell)$ for a price of at least $c_i$ and loses at $W(s_r)$.} In this case, seller $i$ prefers the outcome of winning at $W(s_\ell)$ for a price that would give them non-negative utility rather than losing at $W(s_r)$, which would give them zero utility. As a result, they would accept the price decrease and $W(s)=W(s_\ell)$. 

Since $W(s)=W(s_\ell)$, to prove that the lemma holds for $s$ as well, it suffices to show that
\begin{equation*}
W(s_\ell)\in \arg\max_X v(X) - \sum_{i \in X} \hat{p}_i(s),  
\end{equation*}
and that $\tau$ would choose this allocation in case of ties. We can once again verify this is true using the fact that $\hat{p}_j(s)=\hat{p}_j(s_\ell)$ for every seller $j$, since no additional freezing took place between $s$ and $s_\ell$. Therefore, our inductive assumption that the lemma holds for node $s_\ell$ directly implies that it also holds for $s$.
\end{proof}

Using Lemma~\ref{lem:induct_eff}, we can now verify that the descending auction is guaranteed to be efficient in any subgame-perfect equilibrium. 
\begin{theorem}\label{thm:efficient-descending-bafo}
The allocation induced by any descending auction with BAFO in any subgame perfect equilibrium is always efficient.
\end{theorem}
\begin{proof}
Let $s$ be the root node of the game tree. Since no seller has had a chance to freeze at that point, i.e., $F(s)=\emptyset$, we have $\hat{p}_i=c_i$ for all $i$. Using Lemma~\ref{lem:induct_eff} for $s$ we can conclude that the allocation $W(s)$ resulting from a subgame-perfect equilibrium satisfies
\begin{equation*}    
W(s_\ell)\in\arg\max_X v(X) - \sum_{i \in X} c_i(s),
\end{equation*}
which implies that it is an efficient allocation.
\end{proof}

\begin{remark} Similar to Remark~\ref{remark:nyb-order-independent}, the proof Theorem \ref{thm:efficient-descending-bafo} shows that sellers do not need to know the order in which the buyer will approach the sellers to compute their SPE strategy.
\end{remark}

\begin{remark} If the valuation function satisfies gross substitutes, then the final allocation never includes a seller that froze their price. By the definition of substitutability, if an item is not demanded at a given price vector, it is not demanded at any vector where every other price is weakly smaller. Hence, for the special case of substitutes, our auction behaves exactly like the procedure of Kelso and Crawford \cite{KC82}.
\end{remark}

\section{The Cost of Descending Auctions with BAFO}\label{sec:prices}
Having shown that the subgame perfect equilibria of any descending auction with BAFO always yield efficient allocations, we now focus on the prices vectors that they induce. Our next result shows that these price vectors are not unique; not even with respect to the buyer's total cost (the sum of the winners' prices). In fact, we show that even for the special class of anonymous valuations, which depend only on the number of sellers rather than who these sellers are, the total cost of the buyer can vary by a factor that grows linearly with the number of sellers.

\begin{theorem}\label{thm:cost}
The total cost of the buyer in SPE outcomes of (slightly) different descending auctions with BAFO over the same instance can vary by a factor $n/2$, even for the special case where the valuation function of the buyer is anonymous. On the other hand, if the anonymous valuation function is also weakly concave, then the price vector in every SPE is unique (a threshold price for all winners).
\end{theorem}
\begin{proof} 
Consider an instance where the buyer's value for a set of sellers $Q\subseteq N$ is $v(Q)=|Q|$ if $|Q|\leq n-2$, but if $|Q|=n-1$ it is $v(Q)=n-2$, and if $|Q|=n$ it is $v(Q)=n-1$. In this instance, the cost of every seller is zero. Therefore, the allocation $x^*$ that maximizes the social welfare buys from all the sellers, leading to social welfare equal to $n-1$. As our previous results show, every SPE of the Name-Your-BAFO auction will return this allocation but, as we show here, the resulting prices can very by a lot.

Using the instance above, we consider two executions of the same descending auction with BAFO that only differ with respect to the initial prices. In the first execution, the prices are all initialized at $h=1$, while in the second one they are initialized at $h=1/2$ (which is still much higher than the seller costs). Although one's first guess may be that the former is the one that yields the higher cost for the buyer, somewhat surprisingly, we show that the opposite holds, i.e., starting from lower prices leads to a much higher total cost.

We first consider the execution of the auction that initializes all prices at $h=1$, and we claim that the first bidder who is approached by the auction will choose to directly freeze their price at $1$, and every subsequent bidder will not freeze their price until it reaches $0$, leading to a total cost of $1$ and utility $n-2$ for the buyer in $x^*$. To verify that this is a SPE, assume that one of the subsequent bidders does indeed freeze their price before reaching zero, even though they have observed the fact that the first bidder froze their price at $1$. Let $i$ be the first bidder to do so and let $p_i>0$ be the price at which they freeze. If this happens, note that the total final cost of these two frozen bidders will be $1+p_i>1$, which leads to a contradiction, because it implies that the outcome of the auction will not be $x^*$: dropping these two agents from $x^*$ reduces the total value by $1$ (from $n-1$ to $n-2$) but it reduces the cost by $1+p_i>1$, leading to a higher utility for the buyer. Therefore, if the first bidder freezes at $1$, all subsequent bidders will freeze at $0$. Finally, to verify that it is a SPE strategy for the first bidder to freeze at $1$, note that this is the highest payment that they could hope to receive in this auction and they actually receive it in this outcome.

We now consider the execution of the auction that initializes all prices at $h=1/2$, and we claim that the every bidder who is approached by the auction will choose to directly freeze their price at $1/2$, leading to a total cost of $n/2$ in $x^*$. Note that $x^*$ does maximize the utility of the buyer, providing him with $n/2-1$ (a value of $n-1$ and a total payment of $n/2$), while the utility of any other subset would be at most $n/2-1$. Verifying that this is an SPE is rather straightforward, since every bidder receives the maximum payment that they could hope to receive in this auction.

For the special case of concave valuation functions, note that if the valuation function is concave, then the optimal solution $x^*$ can be derived as follows. Order the sellers in a weakly increasing order of their cost and rename them so that $c_i$ is the $i$-th smallest cost. Then, if we let $v(k)-v(k-1)$ denote the marginal change in the buyer’s value after adding a $k$-th seller to a set of $k-1$ sellers, the optimal solution corresponds to the prefix of sellers in the aforementioned ordering for which $v(k)-v(k-1) \geq c_k$, i.e., their marginal contribution to the buyer’s value is at least as high as their cost. In any SPE of a descending auction with BAFO, every winning seller, i.e., the first k sellers in the ordering, freezes their price at $v(k+1)-v(k)$. It is easy to verify that these prices correspond to a SPE, so we now show that no other price vector could arise in a SPE. Consider any other price vector and note that if one of the smallest $k$ prices, i.e., a price of one of the $k$ winning sellers of $x^*$, is higher than $v(k+1)-v(k)$, then dropping that seller from the winning set would increase the buyer’s utility. On the other hand, if one of these $k$ prices is lower than $v(k+1)-v(k)$, then the corresponding seller could freeze their price sooner while remaining in the winning set.
\end{proof}

\section{Conclusion and Future Directions}
Our results provide a clear positive message regarding the efficiency of the sequential procurement auctions that we propose, and the significant value of providing the bidders with a BAFO strategy. 

A direction for future research that seems intriguing is to develop a deeper understanding regarding the price vectors that can arise with each one of these auctions, and what total cost that they lead to. For instance, is it possible to provide upper bounds on how much the cost of the buyer can vary when the valuation function is submodular? (note that although the construction of Theorem~\ref{thm:cost} uses symmetric functions, these functions are not concave) Alternatively, whenever this cost can vary significantly, how should the buyer choose the auction aiming to minimize the realized cost?

Another interesting direction would be to better understand what information the sellers need to have regarding each other's cost in order to reach (approximately) efficient outcomes.

\bibliographystyle{plainnat}
\bibliography{refs}
\end{document}